\documentclass{article}
\usepackage[totalwidth=13.0cm,totalheight=20.0cm]{geometry}

\usepackage{latexsym,amsthm,amsmath,amssymb,url}

\newtheorem{lemma}{Lemma}

\newtheorem{theorem}{Theorem}

\newcommand{\al}{\alpha}
\newcommand{\MST}{\textsc{Max $r$-Sat AA}}
\newcommand{\MrLT}{\textsc{Max $r(n)$-Lin-2 AA}}
\newcommand{\MrcLT}{\textsc{Max $r$-Lin-2 AA}}
\newcommand{\MLT}{\textsc{Max Lin-2 AA}}
\newtheorem{krule}{Reduction Rule}

\begin{document}

\title{Note on Max Lin-2 above Average}

\author{
Robert Crowston, Gregory Gutin and Mark Jones\\
\small Department of Computer Science\\[-3pt]
\small  Royal Holloway, University of London\\[-3pt]
\small Egham, Surrey TW20 0EX, UK\\[-3pt]
\small \texttt{robert|gutin|markj@cs.rhul.ac.uk}
}
\date{}
\maketitle

\newenvironment{compress}{\baselineskip=10pt}{\par}
\begin{abstract}
\noindent
In the  Max Lin-2 problem we are given a system $S$ of $m$ linear equations
in $n$ variables over $\mathbb{F}_2$ in which Equation $j$ is assigned a positive integral weight $w_j$ for each $j$. We wish to
find an assignment of values to the variables which maximizes the total weight of satisfied equations.
This problem generalizes Max Cut. The expected weight of satisfied equations is $W/2$, where $W=w_1+\cdots +w_m$; $W/2$
is a tight lower bound on the optimal solution of Max Lin-2.

Mahajan et al. (J. Comput. Syst. Sci.  75, 2009) stated the following parameterized version of Max Lin-2: decide whether
there is an assignment of values to the variables that satisfies equations of total weight at least $W/2+k$, where
$k$ is the parameter. They asked whether this parameterized problem is fixed-parameter tractable, i.e., can be solved
in time $f(k)(nm)^{O(1)}$, where $f(k)$ is an arbitrary computable function in $k$ only. Their question remains open, but
using some probabilistic inequalities and, in one case, a Fourier analysis inequality, Gutin et al. (IWPEC 2009) proved that the problem is
fixed-parameter tractable in three special cases.

In this paper we significantly extend two of the three special cases using only
tools from combinatorics. We show that one of our results can be used to obtain a combinatorial proof that another problem from Mahajan et al. (J. Comput. Syst. Sci.  75, 2009),
Max $r$-SAT above the Average, is fixed-parameter tractable for each $r\ge 2.$ Note that Max $r$-SAT above the Average has been already shown to be fixed-parameter
tractable by Alon et al. (SODA 2010), but the paper used the approach of Gutin et al. (IWPEC 2009).

\end{abstract}

\section{Introduction}\label{section:intro}

A \emph{parameterized problem} is a subset $L\subseteq \Sigma^* \times
\mathbb{N}$ over a finite alphabet $\Sigma$. $L$ is
\emph{fixed-parameter tractable} if the membership of an instance
$(x,k)$ in $\Sigma^* \times \mathbb{N}$ can be decided in time
$f(k)|x|^{O(1)}$ where $f$ is a computable function of the
{\em parameter} $k$ only~\cite{DowneyFellows99,FlumGrohe06,Niedermeier06}.
If the nonparameterized version of $L$ (where $k$ is just a part of the input)
is NP-hard, then the function $f(k)$ must be superpolynomial
provided P $\neq$ NP. Often $f(k)$ is moderately exponential,
which makes the problem practically tractable for small values of~$k$.
Thus, it is important to parameterize a problem in such a way that the
instances with small values of $k$ are of interest.

Consider the following well-known problem: given a connected digraph $D=(V,A)$,
find an acyclic subdigraph of $D$ with the maximum number of arcs. We
can parameterize this problem in the standard way by asking whether $D$
contains an acyclic subdigraph with at least $k$ arcs. It is easy to
prove that this parameterized problem is fixed-parameter tractable by
observing that $D$ always has an acyclic subdigraph with at least
$|A|/2$ arcs. Indeed, if $k\le |A|/2$, the answer is {\sc yes} and if $k>|A|/2$ then $|V|\le |A|+1\le 2k.$
In the last case, we can check whether $D$ has an acyclic subdigraph with at least $k$ arcs by generating all $|V|!\le (2k)!$ orderings
of $V$ and constructing subdigraphs of $D$ induced by forward arcs.
This gives an $|A|^{O(1)}(2k)!$-time algorithm. However, this algorithm is impractical as $k>|A|/2$ is
large when $|A|$ is large.

Note that $|A|/2$ is a tight lower
bound on the solution; indeed,  $|A|/2$ is the optimum for all digraphs in which the existence
of an arc implies the existence of the opposite arc. Thus, the following parametrization above a tight lower bound
is more appropriate: decide whether $D$ contains an acyclic subdigraph with at least $|A|/2+k$ arcs.

Mahajan and Raman \cite{MahajanRaman99} were the first to consider
problems parameterized above tight lower bounds (PATLB). They indicated that such parameterizations are often the
only ones of practical value. Mahajan et al. \cite{MahajanRamanSikdar09} proved several
results for problems PATLB, and noted that the parameterized complexity of only a few such problems was
investigated in the literature (partially, because this is often a challenging question) and stated several open questions on the topic. Apart from \cite{MahajanRaman99,MahajanRamanSikdar09}, until very recently there were only three papers on problems PATLB: Gutin et al. \cite{GutinRafieySzeiderYeo07}, Gutin et al.
\cite{GutinSzeiderYeo08} and Heggernes et al. \cite{HeggernesPaulTelleVillanger07}.

 The paper of Mahajan et al. \cite{MahajanRamanSikdar09} triggered
 several recent papers where some of the questions in \cite{MahajanRamanSikdar09} were solved. In particular, Gutin et al. \cite{GutinKimSzeiderYeo09a}
 proved that the above-mentioned maximum acyclic subdigraph problem PATLB is fixed-parameter tractable and so are three special cases of Max Lin-2 PATLB
 (the last problem is defined in the next section). In their proofs, Gutin et al. \cite{GutinKimSzeiderYeo09a} used some probabilistic inequalities and, in one case,
 a Fourier analysis inequality. Using this approach, Alon et al. \cite{AlonEtAl2009a} proved that Max $r$-SAT PATLB  as well as many
 other Boolean Constraint Satisfaction Problems PATLB are fixed-parameter tractable.
 The parameterized complexity of Max $r$-SAT PATLB was one of the central open problems in Mahajan et al. \cite{MahajanRamanSikdar09}.

In this short paper, we significantly extend two of the three special cases of Max Lin-2 PATLB using only tools from combinatorics (see Theorems \ref{thm:main1} and \ref{thm:main2}).
We show that our extensions cannot be obtained by the approach of \cite{GutinKimSzeiderYeo09a}.   We also show that the Boolean Constraint Satisfaction Problems PATLB results of \cite{AlonEtAl2009a} can be proved by combinatorial arguments only,
using one of our results for Max Lin-2 PATLB.

\section{Max $r(n)$-Lin-2 above the Average}\label{sec2}

Consider the following problem for a fixed function $r(n)$. This problem should be called Max $r(n)$-Lin-2 PATLB, if we follow \cite{GutinKimSzeiderYeo09a},
but the new name appears to be clearer and simpler.

\begin{center}
\fbox{~\begin{minipage}{11cm}
  \textsc{Max $r(n)$-Lin-2 above the Average} (or {\MrLT} for short)

  \smallskip

  \emph{Instance:} A system $S$ of $m$ linear equations
in $n$ variables over $\mathbb{F}_2$, where
no equation has more than $r=r(n)$ variables and
Equation $j$ is assigned a positive integral weight $w_j$, $j=1,\ldots ,m$,
and a nonnegative integer $k$.
We will write Equation $j$ in $S$ as $\sum_{i\in \al_j}z_i=b_j,$ where $\al_j\subseteq \{1,2,\ldots ,n\}$ and $|\al_j|\le r.$
  \smallskip

  \emph{Parameter:} The integer $k$.

\smallskip

  \emph{Question:} Is there an assignment of values to the $n$ variables such that the total weight of the satisfied
equations is at least $(W+k)/2$, where  $W=w_1+\cdots +w_m$ ?

\smallskip
\end{minipage}~}
\end{center}

We assume that each of the $n$ variables appears in at at least one equation of $S$ and no equation has an empty left-hand side.

Note that $W/2$ is indeed a tight lower bound
for the above problem, as the expected weight of satisfied
equations in a random assignment is $W/2$, and no assignment of values to the variables
satisfies equations of total weight more than $W/2$ if
$S$ consists of pairs of equations with identical left-hand sides and contradicting right-hand sides.

Consider two reduction rules for {\MrLT} introduced in \cite{GutinKimSzeiderYeo09a}.

\begin{krule}\label{rulerank}
Let $A$ be the matrix of the coefficients of the variables in $S$,
let $t={\rm rank} A$ and let columns $a^{i_1},\ldots ,a^{i_t}$ of $A$ be linearly independent.
Then delete all variables not in $\{z_{i_1},\ldots ,z_{i_t}\}$ from the equations of $S$.
\end{krule}

\begin{krule}\label{rule1}
If we have, for a subset $\al$ of $\{1,2,\ldots ,n\}$, an equation $\sum_{i \in \al} z_i =b'$
with weight $w'$, and an equation $\sum_{i \in \al} z_i =b''$ with weight $w''$,
then we replace this pair by one of these equations with weight $w'+w''$ if $b'=b''$ and, otherwise, by
the equation whose weight is bigger, modifying its
new weight to be the difference of the two old ones. If the resulting weight
is~0, we delete the equation from the system.
\end{krule}

\begin{lemma}\label{mGEn}\cite{GutinKimSzeiderYeo09a}
Let $T$ be obtained from $S$ by Rule~\ref{rulerank} or \ref{rule1}.
Then $T$ is a {\sc yes}-instance if and only if $S$ is a {\sc yes}-instance.
Moreover, $T$ can be obtained from $S$  in time polynomial in $n$ and $m$.
\end{lemma}

We cannot change $S$ using Rule~\ref{rulerank} (Rule \ref{rule1}), $S$ is {\em irreducible} by Rule~\ref{rulerank} (Rule \ref{rule1}).
If $S$ is irreducible by Rule~\ref{rulerank}, we have $n\le m$.
If  $S$ is irreducible by Rule \ref{rule1}, the symmetric difference $\al_j\Delta \al_p\neq \emptyset$ for each pair $j\neq p.$

Consider the following algorithm for {\MrLT}, which is a modification of an algorithm used in \cite{HastadVenkatesh02}. We assume that, in the beginning, no equation or variable in $S$ is marked.

\begin{center}
\fbox{~\begin{minipage}{11cm}
\textsc{Algorithm $\cal A$}

\smallskip
While $S\neq \emptyset$ and less than $k$ equations are marked, do the following:

\begin{enumerate}

					\item For $1\le i\le n$, calculate $\rho_i$, the number of equations in $S$ containing $z_i$.
                                       \item Choose $z_l$ with minimum $\rho_l$ among all variables still in $S$. Mark $z_l.$
                                       \item Choose an arbitrary equation containing $z_l$, $\sum_{i \in \al} z_i =b.$
                                       \item Mark this equation and delete it from $S$.
                                       \item Replace every equation $\sum_{i \in \al'} z_i =b'$ in $S$ containing $z_l$ by
                                       $\sum_{i \in {\al\Delta\al'}} z_i =b''$, where $b''=b+b'.$
                                       \item Apply Rule \ref{rule1}. (As a result, several equations can be of weight 0 and, thus, are deleted from the system.)
                                     \end{enumerate}
\smallskip
\end{minipage}~}
\end{center}
Observe that $\cal A$ runs in polynomial time. We have the following simple yet important property of $\cal A$.


\begin{lemma}\label{lem:marked}
If the input system $S$ is irreducible by Rule \ref{rule1} and algorithm $\cal A$ has marked $k$ equations in $S$, then $S$ is a {\sc yes}-instance.
\end{lemma}
\begin{proof}
Assume that $\cal A$ has marked $k$ equations in the {\em input} system $S$ and let $T$ be the system of equations remained in $S$ after $\cal A$ has stopped.
Observe that for every assignment of values to the variables $z_1,\ldots ,z_n$ that satisfies all marked equations, the operation of Step 5 of $\cal A$
replaces $S$ by an {\em equivalent} system (i.e., both systems have the same difference in weight of satisfied and falsified equations).
Thus, for every such assignment, $S$ is equivalent to $T$ together with the marked equations.
We will show that there is an assignment that satisfies all marked equations and half of equations of $T$ (in terms of weight). This will be sufficient due to the following.
Let $W'$ be the total weight of the marked equations. Then the total weight of the satisfied equations is $W'+(W-W')/2= (W+W')/2\ge (W+k)/2$ since $W'\ge k$ by integrality of the weights.

We can find a required assignment as follows.
We start by finding an assignment of values to the variables in $T$ that satisfies half of equations of $T$ (in terms of weight), using the following algorithm from \cite{HastadVenkatesh02}:  Assign values to the variables sequentially, and after each assignment, perform the obvious algebraic simplifications. When about to assign a value to $x_j$, consider all equations of the form $x_j = b$, for constant $b$. Assign $x_j$ a value satisfying at least half of these equations (in terms of weight).

It remains to assign any values to the variables not in $T$ of the marked equations such that they are all satisfied. This is possible if we
find an assignment that satisfies the last marked equation, then find an assignment satisfying the equation marked before the last, etc. Indeed,
the equation marked before the last contains
a (marked) variable $z_l$ not appearing in the last equation, etc.
\end{proof}

\begin{lemma}\label{lemma:yesinst}
If an instance of {\MrLT} is irreducible by Rule \ref{rule1} and its number of variables $n \ge 2^{k}r(n)$, then it is a {\sc yes}-instance.
\end{lemma}
\begin{proof}
Let $\rho_{l_t}$ be the $\rho_l$ picked in step 2 of Iteration $t$ of algorithm $\cal A$, and let $R_t$ be the maximum number of variables in any equation in $S$ at Iteration $t$.
Observe that $R_1$ = $r$, and that $R_{t+1}\le 2R_t$. Thus, $R_t \le 2^{t-1}r$.
In  Iteration $t$ of $\cal A$, by minimality of $\rho_{l_t}$, at  most $(2\rho_{l_t}-1)R_t/\rho_{l_t}<2R_t$  variables will be removed from the system.
Thus, the total number of variables completely deleted from the system after $k-1$ iterations is less than $\sum_{t=1}^{k-1} 2R_t \le \sum_{t=1}^{k-1} 2^t r<2^{k}r.$ So, if $2^{k}r\le n$ then Iteration $k$ is possible, and hence, by Lemma \ref{lem:marked}, we have a {\sc yes}-instance.
\end{proof}

\begin{theorem}\label{thm:main1}
If an instance of {\MrLT} is irreducible by Rule \ref{rule1} and  $r(n)=o(n)$, then {\MrLT} is fixed-parameter tractable.
\end{theorem}
\begin{proof}
Let $r=o(n)$. By Lemma \ref{lemma:yesinst}, if $n \ge 2^{k}r$, then we have a {\sc yes}-instance. Otherwise, $n < 2^{k}r$ and so
$n\le g(k)$ for some function $g(k)$ depending on $k$ only. In the last case, in time $O(m^{O(1)}2^{g(k)})$ we can check whether our instance is a {\sc yes}-instance.
\end{proof}

Gutin et al. \cite{GutinKimSzeiderYeo09a} prove that {\MrLT} is fixed-parameter tractable for $r=O(1).$ Using the method of \cite{GutinKimSzeiderYeo09a} one can only extend this result to $r= o(\log m)$.  If $r= o(\log m)$ then $r=o(n)$ (since $m<2^n$  by Rule \ref{rule1}) and, thus, {\MrLT} is fixed-parameter tractable by Theorem \ref{thm:main1}. However, if $r= \Omega(\log m)$ and $r=o(n)$ then
{\MrLT} is fixed-parameter tractable by Theorem \ref{thm:main1}, but this result cannot be obtained using the method of \cite{GutinKimSzeiderYeo09a}.

Let $\rho=\max_{1\le i\le n}\rho_i.$ Let {\MLT} be  {\MrLT} with $r(n)=n$.

\begin{theorem}\label{thm:main2}
Let the input system $S$ be irreducible by Rules \ref{rulerank} and \ref{rule1} and let $S$ have $m$ equations.
If $\rho=o(m)$,  then {\MLT} is fixed-parameter tractable.
\end{theorem}
\begin{proof}
Let $\rho=o(m)$. Apply algorithm $\cal A$ (for this theorem, there is no need to do Step 1 or select the $z_l$ with minimum $\rho_l$ on Step 2; we can arbitrarily choose any $z_l$ still in $S$).
We will show that after $k-1$ iterations at most $2\rho (k-1)$ equations
have been deleted. Let $Q$ be the set of equations that, at the beginning, contain at least
one
of $z_{l_1}, \ldots, z_{l_{k-1}}$, where $z_{l_1}, \ldots,
z_{l_{k-1}}$ are the variables marked in the first $k-1$
iterations. Note that $|Q| \le \rho (k-1)$. An equation not in $Q$ is
only deleted if there exists an equation in $Q$ such that, after some
applications of the symmetric difference operation of Step 5, the two
equations have the same left-hand side. Furthermore, observe that each
equation in $Q$ can only ever have the same left-hand side as at most one
equation not in $Q$. So the number of equations removed is at most
$2|Q|\le 2\rho (k-1)$.
Observe that either $2\rho (k-1)<m$ in which case Iteration $k$ is possible and we can apply Lemma \ref{lem:marked}, or
$m\le 2\rho (k-1)$ and $m\le f(k)$ for some
function $f(k)$ depending on $k$ only. If $m\le f(k)$, $n\le m\le
f(k)$
and in time $O(m^{O(1)}2^{f(k)})$ we can check whether our instance is
a {\sc yes}-instance.
\end{proof}

Using the approach of \cite{GutinKimSzeiderYeo09a} it is easy to show that if $\rho=o(\sqrt{m})$  then {\MLT} is fixed-parameter tractable and this cannot be extended even to the case $\rho=\Theta(\sqrt{m}).$ Thus, Theorem \ref{thm:main2} provides a much stronger result.

\section{Boolean Constraint Satisfaction Problems above Average}\label{sec3}

The aim of this section is to show that in the proofs of the main results of \cite{AlonEtAl2009a} for a wide family of Boolean Constraint Satisfaction Problems above the Average, Lemma \ref{lemma:yesinst} can replace probabilistic and Fourier analysis inequalities. As a result, the proofs become purely combinatorial and slightly simpler. Alon et al. \cite{AlonEtAl2009a} provide all details for {\MST} (defined below) only and comment that basically the same arguments can be used for a wide class of Boolean Constraint Satisfaction Problems above the Average. Thus, we restrict ourselves  to {\MST} only as basically the same arguments can be used for the wide class of Boolean Constraint Satisfaction Problems above the Average.

Let $r(\ge 2)$ be a constant.

\begin{center}
\fbox{~\begin{minipage}{11cm}
  \textsc{Max $r$-Sat above the Average} (or {\MST} for short)

  \smallskip

  \emph{Instance:} A pair $(F,k)$ where $F$ is a multiset of $m$
  clauses, each of size $r$;
  $F$ contains only variables $x_1, x_2, \ldots ,x_n$, and $k$ is a nonnegative integer.

  \smallskip

  \emph{Parameter:} The integer $k$.

\smallskip

  \emph{Question:} Is there a truth assignment to the $n$ variables such that the total number of satisfied clauses is at least $E+k2^{-r}$, where $E=m(1-2^{-r})$, the average number of satisfied clauses?

\smallskip
\end{minipage}~}
\end{center}

Let $F$ contain clauses $C_1,\ldots ,C_m$
in the variables $x_1, x_2, \ldots ,x_n$. We may assume that $x_i \in \{-1,1\}$, where $-1$ corresponds to {\sc true}.
For $F$, consider a polynomial
\[
X=\sum_{j=1}^m (1-\prod_{x_i\in C_j}(1+\epsilon_i x_i)),
\]
where $\epsilon_i\in \{-1,1\}$ and $\epsilon_i=1$ if and only if $x_i$
is in~$C$.

\begin{lemma}\label{lem:MrSTlem}\cite{AlonEtAl2009a}
The answer to {\MST} is {\sc yes} if and only if there exists an
assignment for $x_1, x_2, \ldots ,x_n$ for which $X \ge k$.
\end{lemma}

\begin{theorem}\label{thm:MrST}
The problem {\MST} is fixed-parameter tractable for each constant $r\ge 2$.
\end{theorem}

\begin{proof}
Given a {\MST} instance, define the polynomial $X$ of degree at most $r$
as above. After algebraic simplification $X=X(x_1, x_2, \ldots ,x_n)$
can be written as
$\label{Xeq}X=\sum_{I\in {\cal S}}X_I,$ where $X_I=c_I\prod_{i\in
I}x_i$, each $c_I$ is a nonzero integer and $\cal S$ is a family of
nonempty subsets of $\{1,\ldots ,n\}$ each with at most $r$
elements. Thus, $X$ is a polynomial of degree at most $r$.

Now define an instance {\MrcLT} with the variables
$z_1, z_2, \ldots ,z_n$ as follows. For each
nonzero term $c_I \prod_{i \in I} x_i$  consider the linear equation
$\sum_{i \in I} z_i =b$, where $b=0$ if $c_I$ is positive, and $b=1$
if $c_I$ is negative, and assign this equation the
weight $w_I=|c_I|$.  It is easy to check
that this system of equations has an assignment $z_i$ satisfying equations of total weight at least
$[\sum_{I \in {\cal S}} w_I+k]/2$  if and only if
there are $x_i \in \{-1,1\}$ so that $X(x_1,x_2, \ldots ,x_n)
\geq k$. This is shown by the transformation $x_i=(-1)^{z_i}$. Let $n'$ be the number of variables in the instance of {\MrcLT}; clearly $n'\le n.$ Observe that $|S|\le n^r$.

By Lemma \ref{lemma:yesinst}, if $n'> 2^{k}r$, then we have a {\sc yes}-instance of {\MrcLT} and, thus, by Lemma \ref{lem:MrSTlem}, the answer to {\MST} is {\sc yes}.
If $n' \le 2^{k}r$ then we can find the maximum of $X$ by using all assignments in time  $|S|^{O(1)}2^{n'}=n^{O(r)}2^{r2^k}$ and apply Lemma \ref{lem:MrSTlem} to check whether the answer to {\MST} is {\sc yes}.

It remains to observe that the instance of {\MrcLT} can be constructed in time $(m2^r)^{O(1)}.$
\end{proof}

\medskip

\paragraph{Acknowledgments}
We are thankful to Noga Alon for helpful discussions. Research of Gutin and Jones was supported in part by an EPSRC
grant. Research of Gutin was also supported
in part by the IST Programme of the European Community, under the
PASCAL 2 Network of Excellence.

\urlstyle{rm}

\end{document}